\documentclass[a4paper,UKenglish,cleveref, autoref]{lipics-v2019}

\usepackage{amsmath, amssymb, amscd, amsthm, amsfonts}
\usepackage{graphicx}
\usepackage{hyperref}
\usepackage{mathtools,xparse,bigstrut}
\usepackage{physics}
\usepackage{algorithm}
\usepackage{algpseudocode}
\usepackage{float}

\newfloat{algorithm}{t}{lop}

\DeclarePairedDelimiter{\absl}{\lvert}{\rvert}
\DeclarePairedDelimiter{\nor}{\lVert}{\rVert}
\NewDocumentCommand{\normL}{ s O{} m }{%
  \IfBooleanTF{#1}{\nor*{#3}}{\nor[#2]{#3}}_{2}%
}
\NewDocumentCommand{\normF}{ s O{} m }{%
  \IfBooleanTF{#1}{\nor*{#3}}{\nor[#2]{#3}}_{F}%
}
\newcommand{\Conv}{{\normalfont \textrm{Conv}}}
\newcommand{\row}{{\normalfont \textrm{row}}}
\newcommand{\col}{{\normalfont \textrm{col}}}
\newcommand{\poly}{{\normalfont \textrm{poly}}}
\newcommand{\polylog}{{\normalfont \textrm{polylog}}}
\newcommand{\diag}{{\normalfont \textrm{diag}}}
\newcommand{\inv}{{\normalfont \textrm{inv}}}

\newcommand{\thres}{{\normalfont \textrm{th}}}
\newcommand{\Complex}{\mathbb{C}}

\newcommand{\Realnn}{\mathbb{R}_{\ge 0}}
\newcommand{\ceil}[1]{\left\lceil #1 \right\rceil}

\makeatletter
\newtheorem*{rep@theorem}{\rep@title}
\newcommand{\newreptheorem}[2]{%
\newenvironment{rep#1}[1]{%
 \def\rep@title{#2 \ref{##1}}%
 \begin{rep@theorem}}%
 {\end{rep@theorem}}}
\makeatother

\newreptheorem{theorem}{Theorem}

\title{Quantum-Inspired Classical Algorithms for Singular Value Transformation} 

\author{Dhawal Jethwani}{
Indian Institute of Technology (BHU), Varanasi, India}{dhawal.jethwani.cse15@iitbhu.ac.in}{[orcid]}{}

\author{Fran{\c c}ois Le Gall}{
Nagoya University, Japan 
}{legall@math.nagoya-u.ac.jp}{[orcid]}{}

\author{Sanjay K. Singh}{
Indian Institute of Technology (BHU), Varanasi, India}{sks.cse@iitbhu.ac.in}{[orcid]}{}

\authorrunning{D. Jethwani, F. Le Gall and S.\,K. Singh}

\Copyright{Dhawal Jethwani, Fran{\c c}ois Le Gall and Sanjay K. Singh}



\category{}

\acknowledgements{The authors are grateful to Andr{\'{a}}s Gily{\'{e}}n for discussions and comments about the manuscript. FLG was partially supported by JSPS KAKENHI grants Nos.~JP15H01677, JP16H01705, JP16H05853, JP19H04066 and by the MEXT Quantum Leap Flagship Program (MEXT Q-LEAP) grant No.~JPMXS0118067394. Part of this work has been done when DJ was visiting Kyoto University.}

\nolinenumbers 

\hideLIPIcs  

\EventEditors{Javier Esparza and Daniel Kr{\'a}l'}
\EventNoEds{2}
\EventLongTitle{45th International Symposium on Mathematical Foundations of Computer Science (MFCS 2020)}
\EventShortTitle{MFCS 2020}
\EventAcronym{MFCS}
\EventYear{2020}
\EventDate{August 24--28, 2020}
\EventLocation{Prague, Czech Republic}
\EventLogo{}
\SeriesVolume{170}
\ArticleNo{47}

\begin{document}

\maketitle

\begin{abstract}
A recent breakthrough by Tang (STOC 2019) showed how to ``dequantize'' the quantum algorithm for recommendation systems by Kerenidis and Prakash (ITCS 2017). The resulting algorithm, classical but ``quantum-inspired'', efficiently computes a low-rank approximation of the users' preference matrix. Subsequent works have shown how to construct efficient quantum-inspired algorithms for approximating the pseudo-inverse of a low-rank matrix as well, which can be used to (approximately) solve low-rank linear systems of equations. In the present paper, we pursue this line of research and develop quantum-inspired algorithms for a large class of matrix transformations that are defined via the singular value decomposition of the matrix. In particular, we obtain classical algorithms with complexity polynomially related (in most parameters) to the complexity of the best quantum algorithms for singular value transformation recently developed by Chakraborty, Gily{\'{e}}n and Jeffery (ICALP 2019) and Gily{\'e}n, Su, Low and Wiebe (STOC 2019).  
\end{abstract}

\section{Introduction}
\label{sec-intro}
\noindent{\bf Background.}
One of the most celebrated quantum algorithms discovered so far is the HHL  algorithm~\cite{HHL09}. This quantum algorithm solves a system of linear equations of the form $Ax=b$, where $A$ is an $n\times n$ matrix and $b$ is an $n$-dimensional vector, in time polynomial in $\log n$ when the matrix $A$ is sufficiently sparse and well-conditioned. This is exponentially better that the best known classical algorithms, which run in time polynomial in $n$ (see also \cite{AA12,CKS17,CJS13,WZP18} for improvements and relaxations of the assumptions). There are nevertheless two significant caveats. First, the input should be given in a way that allows very specific quantum access. In particular, the HHL algorithm requires the ability to efficiently create a quantum state proportional to~$b$. The second, and main, caveat is that the output of the HHL algorithm is not the solution $x$ of the linear system (which is an $n$-dimensional vector) but only a $O(\log n)$-qubit quantum state proportional to this vector. While measuring this quantum state can give some meaningful statistics about the solution~$x$, this naturally does not give enough information to obtain the whole vector $x$. In this perspective, the HHL algorithm does not explicitly solve the system of equations, but instead enables sampling from the solution, in a very efficient way. 

There have been several proposals to apply the HHL algorithm (and one of its core components, phase estimation) to linear-algebra based machine learning tasks, leading for instance to the discovery of quantum algorithms for principal component analysis (PCA)~\cite{LMR14} and quantum support vector machine~\cite{RML14}. We refer to~\cite{BWP17} for a recent survey on this field called quantum machine learning. One of the most convincing applications of quantum algorithms to machine learning has been speeding up recommendation systems~\cite{KPQ16}. In machine learning, recommendations systems are used to predict the preferences of users. From a mathematical perspective, the core task in recommendation systems can be modeled as follows: given an $m\times n$ matrix $A$ (representing the preferences of $m$ users) and an index $i\in[m]$ (representing one specific user), sample from the $i$-th row of a low-rank approximation of $A$. Kerenidis and Prakash \cite {KPQ16} showed how to adapt the HHL algorithm to solve this problem in time polynomial in $\log(mn)$, which was exponentially better than the best known classical algorithms for recommendation systems. 

Similarly to the HHL algorithm, the quantum algorithm from \cite {KPQ16} works only under the assumption that the input is stored in an appropriate structure (called ``Quantum Random-Access Memory'', or ``QRAM'') that allows specific quantum access. Very recently, Tang~\cite{TAN19} has shown that assuming that the input is stored in a classical data structure that allows $\ell^2$-norm sampling access (i.e., allows sampling rows with probability proportional to their $\ell^2$-norm), $\polylog (mn)$-time classical algorithms for recommendation systems can be designed as well. This results eliminates one of the best examples of quantum speedup for machine learning. The paper \cite{TAN19} also introduced the term ``quantum-inspired algorithms'' to refer to such classical algorithms obtained by ``dequantizing'' quantum algorithms. 

More quantum-inspired algorithms have soon been developed: Tang \cite{TAN18B} first showed how to construct classical algorithms for PCA that essentially match the complexity of the quantum algorithm for PCA from \cite{LMR14} mentioned above. Gily{\'e}n, Lloyd and Tang~\cite{GLT18} and, independently, Chia, Lin and Wang \cite{CHI18} have shown how to obtain new classical algorithms for solving linear systems of equations, which also essentially match the complexity of the quantum algorithms when the input matrix has low-rank (see below for details). We also refer to \cite{ADL19} for a discussion of the performance of these quantum-inspired algorithms in practice.

\vspace{2mm} 
 
\noindent{\bf Singular value transformation.}
The Singular Value Decomposition (SVD) of a matrix $M\in \Complex^{m\times n}$ is a factorization of the form $M=U\Sigma V^{*}$ where $U\in \Complex^{m\times m}$ and $V\in \Complex^{n\times n}$ are unitary matrices and $\Sigma$ is a $m\times n$ diagonal matrix with $\min(m,n)$ non-negative real numbers on the diagonal, where $V^\ast$ denotes the complex-conjugate transpose of V. A crucial property is that this decomposition exists for any complex matrix. Given a function $f\colon\Realnn\to\Realnn$, the singular value transformation associated with~$f$, denoted~$\Phi_f$, is the function that  maps the matrix $M=U\Sigma V^{*}$ to the matrix $\Phi_f(M)=U\ \Sigma_f V^{*}$ where $\Sigma_f$ is the diagonal matrix obtained from $\Sigma$ by replacing each diagonal entry $\sigma$ by $f(\sigma)$. We refer to Definition~\ref{def:SVT} in Section \ref{sec-pre} for more details.

An important example is obtained by taking the ``pseudo-inverse'' function $\inv\colon \Realnn\to\Realnn$ such that $\inv(x)=1/x$ if $x>0$ and $\inv(0)=0$. Solving a linear system of equations $Ax=b$ corresponds\footnote{Indeed, one solution is given by $x=A^+ b$, where $A^+$ represents the Moore-Penrose pseudo-inverse of the matrix $A$ (or simply the inverse when $A$ is invertible). It is easy to check that $A^+=\Phi_\inv(A^*).$} to calculating (or approximating) the vector $\Phi_\inv(A^\ast)b$. If all the singular values of $A$ are between $1/\kappa$ and 1, for some value $\kappa$, the quantum-inspired algorithms from~\cite{CHI18,GLT18} solve this task in time $\poly\left(k_A, \kappa,\normF{A}, 1/\epsilon,\log(mn)\right)$, where~$k_A$ denotes the rank of~$A$, $\normF{A}$ denotes the Frobenius norm of $A$ and $\epsilon$ denotes the approximation error.\footnote{The term $\log(mn)$ represents the time complexity of implementing sampling and query operations (see Proposition \ref{tanm} in Section \ref{sub:data}), which we also include in the complexity.} One crucial point here is that the dependence on the dimensions of the matrix is only poly-logarithmic. Another important point is that the best known quantum algorithms (see~\cite{CGJ19,GLT18}) enable $\ell^2$-norm sampling from the output in time $O\left(\kappa\normF{A}\polylog(mn/\epsilon)\right)$ in the QRAM input model. This means that, except for the dependence in $\epsilon$, for low-rank matrices the classical running time is polynomially related to the quantum running time.

The core computational problem in recommendation systems can also be described as approximating the $i$-row of the matrix $\Phi_{\thres}(A)$ for the threshold function $\thres\colon \Realnn\to\Realnn$ such that $\thres(x)=x$ if $x\ge\sigma$ and $\thres(x)=0$ otherwise (for some appropriate threshold value~$\sigma$). This corresponds to approximating the vector $\Phi_{\thres}(A^\ast)b$ where $b$ is the vector with~$1$ in the $i$-th coordinate and zero elsewhere. Ref.~\cite{TAN19} shows how to solve this problem in time $\poly\left(\normF{A}/\sigma, 1/\epsilon,\log(mn)\right)$. (For the value $\sigma$ chosen for recommendation systems, the term $\normF{A}/\sigma$ becomes an upper bound on the rank of a low-rank approximation of $A$.)

\vspace{2mm}

\noindent{\bf Our results.}
In this paper we significantly extend the class of functions for which the singular value transformation can be efficiently computed by ``quantum-inspired'' classical algorithms. 
The formal and most general statements of our results are given in Section \ref{sec:main}. For the sake of readability, in this introduction we only describe our results for a restricted (but still very general) class of ``smooth'' functions. Let $\Realnn$ and $\mathbb{R}_{>0}$ denote the sets of non-negative numbers and positive numbers, respectively. We say below that a function $f\colon\Realnn\to\Realnn$ is ``smooth'' if~$f$ is differentiable in $\mathbb{R}_{>0}$ and the following condition holds: for any $\alpha,\beta\ge 1$, over the interval $[1/\alpha,\beta]$ the maximum values of $f$ and its derivative $f'$ can be upper bounded by a polynomial function of $\alpha$ and $\beta$. We are mostly interested in functions such that $f(0)=0$ since typically we do not want the transformation to increase the rank.

Our main results are the following two theorems (we refer to Section \ref{sec:main} for the formal versions).\footnote{These informal versions can be derived from the formal versions given in Section \ref{sec:main} by observing that $\kappa_{2}/\normL{A}\le \kappa$ if all the singular values of $A$ are between $1/\kappa$ and 1. The smoothness condition implies that both $\Omega$ and $\phi$ are upper bounded by a polynomial of $\kappa$ and $\normF{A}$. Note that for Theorem \ref{th2} we actually need an additional smoothness condition expressing that the minimum value of $f$ cannot be too small as well (see the term $\omega$ in the formal version of Theorem \ref{th2}).}

\begin{theorem}[Informal Version]\label{th1}
    Let $f\colon\Realnn\to\Realnn$ be any smooth function such that $f(0)=0$.
    For any sufficiently small $\epsilon>0$, there exists a classical algorithm that has sampling access to a matrix $A \in \Complex^{m\times n}$ with singular values in  $[1/\kappa, 1]$ and to a non-zero vector $b\in \Complex^{m}$, receives as input an index $i\in [n]$, outputs with high probability an approximation of the $i$-th coordinate of the vector $\Phi_{f}(A^{*})b$ with additive error~$\epsilon$, and has $\poly\left(\kappa, \normF{A}, 1/\epsilon,\log(mn)\right)$ time complexity.
\end{theorem}

\begin{theorem}[Informal Version]\label{th2}
    Let $f\colon\Realnn\to\Realnn$ be any smooth function such that $f(0)=0$ and $f(x)>0$ for all $x>0$.
    For any sufficiently small $\epsilon>0$, there exists a classical algorithm that has sampling access to a matrix $A \in \Complex^{m\times n}$ with singular values in  $[1/\kappa, 1]$ and to a non-zero vector $b\in \Complex^{m}$, and $\ell^2$-samples with high probability from a distribution $\epsilon$-close in total variation distance to the distribution associated with the vector $\Phi_{f}(A^{*})b$, and has $\poly\left(\kappa, \normF{A}, 1/\epsilon,\log(mn)\right)$ time complexity.
\end{theorem}

Note that instead of stating our results for the transformation $\Phi_{f}(A)$ we state them for the transformation $\Phi_{f}(A^{*})=(\Phi_{f}(A))^*$ in Theorems \ref{th1} and \ref{th2}. The reason is that this simplifies the presentation of our algorithms and makes the comparison with prior works easier.

Theorems \ref{th1} and \ref{th2} show that under the same assumptions (namely, sampling access to the input) and similar requirements for the output (i.e., outputting one coordinate of $\Phi_{f}(A^{*})b$ or sampling from the associated distribution) as the prior works on quantum-inspired algorithms, we can efficiently compute classically the singular value transformation for any smooth enough function. This extends the results from \cite{CHI18,GLT18,TAN19} and significantly broadens the applicability of quantum-inspired algorithms.

Fast quantum algorithms have been constructed in recent works \cite{CGJ19,GSL19} for singular value transformations. For the class of smooth functions we consider, the quantum running time obtained would be $O\left(\poly\left(\kappa,\normF{A},\log(mn/\epsilon)\right)\right)$ in the QRAM input model. Our results thus show that except possibly for the dependence on $\epsilon$, we can again obtain classical algorithms with running time polynomially related to the quantum running time.

\vspace{2mm} 

\noindent{\bf Overview of our approach.}
We use the same sampling methods as in \cite{ADL19,CHI18,FKV04,GLT18,TAN19}: we first sample $r$ rows from the input matrix $A\in \Complex^{m\times n}$ according to probability proportional to the row norms, which gives (after normalization) a matrix $S\in\Complex^{r\times n}$. We then do the same with matrix $S$, this time sampling $c$ columns, which gives (after normalization) a matrix $W\in \Complex^{r\times c}$. The analysis of this process, which has been done in the seminal work by Frieze, Kannan and Vempala \cite{FKV04}, shows that with high probability we have $A^*A\approx S^*S$ and $SS^*\approx WW^\ast$ when $r$ and $c$ are large enough (but still much smaller than $m$ and $n$). Since $W$ is a small matrix, we can then afford to compute its SVD. 

The main contribution of this paper is the next step (and its analysis). We show how to use the SVD of the matrix $W$ in order to compute the singular value transformation $\Phi_f$. Using the SVD of $W$, we first compute the matrices $\Phi_{\inv}(W)$, $\Phi_{\inv}(W^*)$ and $\Phi_{f}(W)$. We then compute the matrix $P'=\Phi_{\inv}(W)\Phi_{f}(W^{*})\Phi_{\inv}(W)\Phi_{\inv}(W^{*})\in\Complex^{r\times r}$. This matrix $P'$ is the output of Algorithm \ref{alg1} presented in Section \ref{subsec:alg1}. Our central claim is the following:
\begin{equation}
    S^\ast P' S A^\ast\approx \Phi_f(A^\ast).
    \label{eq:central}
\end{equation}

Proving (\ref{eq:central}) and quantifying the quality of the approximation is our main technical contribution. This is done in Proposition \ref{ovec} (which itself relies on several lemmas proved in Sections~\ref{sub1} and \ref{subsec:alg1}). Finally, using similar post-processing techniques as in prior works \cite{CHI18,TAN19}, from the output~$P'$ of Algorithm \ref{alg1} we can efficiently approximate coordinates of $\Phi_f(A^\ast)b$ and sample from $\Phi_f(A^\ast)b$. This post-processing is described in Algorithms \ref{alg2} and \ref{alg3} in Section \ref{sub3}.

We now give an outline of the main ideas used to establish (\ref{eq:central}). The basic strategy is to exploit the relations $A^*A\approx S^*S$ and $SS^*\approx WW^\ast$ mentioned above.  Our first insight is to define the function $h\colon \Realnn\to\Realnn$ such that $h(x)=f(\sqrt{x})/\sqrt{x}$ if $x>0$ and $h(0)=0$, and observe that $\Phi_f(A^\ast)=\Phi_h(A^\ast A)A^\ast$. We then prove, in Lemma~\ref{fpsd}, that $A^*A\approx S^*S$ implies $\Phi_h(A^\ast A)\approx \Phi_h(S^\ast S)$. The next natural step would be to relate $\Phi_h(S^\ast S)$ and $\Phi_h(W^\ast W)$, but this cannot be done directly since the only guarantee is $SS^*\approx WW^\ast$, and not $S^*S\approx W^*W$. Instead, we observe that $\Phi_h(S^*S)=S^\ast P S$, where $P=\Phi_\inv(S)\Phi_f(S^*)\Phi_\inv(S)\Phi_\inv(S^*)$. Since $\Phi_\inv(S)\Phi_f(S^\ast)=\Phi_h(SS^*)$ and $\Phi_\inv(W)\Phi_f(W^\ast)=\Phi_h(WW^*)$, and since we can show that $\Phi_h(SS^*)$ is close to $\Phi_h(WW^*)$ using Lemma \ref{fpsd}, we are able to prove that $P\approx P'$ (this is proved in Lemma \ref{ppcl}). To summarize, we have
$
S^\ast P' S A^\ast\approx 
S^\ast P S A^\ast=
\Phi_h(S^*S)A^\ast\approx
\Phi_h(A^*A)A^\ast=
\Phi_f(A^\ast),
$
as needed.

\vspace{2mm}

\noindent{\bf Related independent work.}
Independently from our work, Chia, Gily{\'e}n, Li, Lin, Tang and Wang simultaneously derived similar results \cite{CGLLTW19}. They additionally provide general matrix arithmetic primitives for adding and multiplying matrices having sample and query access, and recover known dequantized algorithms. They also show how to use these results on the singular value transformation to obtain new quantum-inspired algorithms for other applications, including Hamiltonian simulation and discriminant analysis.
\section{Preliminaries}
\label{sec-pre}
\subsection{Notations and conventions}
\label{subsec-not}
\noindent{\bf General notations.}
In this paper we use the notation $[n]=\{1,.....,n\}$ for any integer $n\ge 1$. 
For any set $S$ we denote $\Conv(S)$ the convex hull of~$S$.

Given a matrix $M\in\Complex^{m\times n}$, we use $M_{(i,.)}\in \Complex^{1\times n}$, $M_{(.,j)}\in \Complex^{m\times 1}$ and $M_{(i,j)}\in \Complex$ to denote its $i$-th row, its $j$-th column and its $(i,j)$-th element, respectively. The complex-conjugate transpose or Hermitian transpose of a matrix $M\in \Complex^{m\times n}$ (or a vector $v\in \Complex^{n}$) is denoted as $M^{*}$ (and $v^{*}$, respectively). The notations $\normF{M}$ and $\normL{M}$ represent the Frobenius and spectral norm, respectively. Note that $\normL{M}\le\normF{M}$ for any $M$. For a vector $v\in \Complex^{n}$, we denote $\nor{v}$ the $\ell^{2}$ norm of the vector. In this paper we will use several times the following standard inequalities that hold for any vector $v\in \Complex^{n}$ and any matrices $M\in \Complex^{n\times m}$ and $N\in \Complex^{m\times p}$:
\begin{equation}\label{vnor}
    \nor{Mv}\leq \normL{M}\nor{v},
    \hspace{7mm}
    \normF{MN}\leq \normL{M}\normF{N}, 
    \hspace{7mm}
    \normF{MN}\leq \normF{M}\normL{N}.
\end{equation} 

For a non-zero vector $v\in \Complex^{n}$, let $\mathcal{P}_{v}$ denote the probability distribution on $[n]$ where the probability of choosing $i\in [n]$ is defined as $\mathcal{P}_{v}(i) = \frac{\absl*{v_{i}}^{2}}{\nor{v}^{2}}$. For two vectors $v$ and $w$, the total variation distance between distributions $\mathcal{P}_{v}$ and $\mathcal{P}_{w}$ is defined as $\nor{\mathcal{P}_{v}-\mathcal{P}_{w}}_{TV} = \frac{1}{2}\sum_{i=1}^{n}\absl*{\mathcal{P}_{v}(i)-\mathcal{P}_{w}(i)}.$

We will use the following easy inequality (see for instance \cite{CHI18,TAN19} for a proof): for any two vectors $v, w\in \Complex^{n}$, 
\begin{equation}\label{eq:dist}
\nor{\mathcal{P}_{v}-\mathcal{P}_{w}}_{TV}\leq \frac{2\nor{v-w}}{\nor{v}}.
\end{equation}\vspace{2mm}

\noindent{\bf Singular Value Decomposition.}
The Singular Value Decomposition (SVD) of a matrix $M\in \Complex^{m\times n}$ is a factorization of the form $M=U\Sigma V^{*}$ where $U\in \Complex^{m\times m}$ and $V\in \Complex^{n\times n}$ are unitary matrices and $\Sigma$ is an $m\times n$ diagonal matrix with $\min(m,n)$ non-negative real numbers, in non-increasing order, down the diagonal. The columns of $U$ and $V$ represent the left and right singular vectors, respectively. Each entry of this diagonal matrix is a singular value of matrix $M$. 
A crucial property is that a SVD exists for any complex matrix. 

We can also write the SVD of a matrix as
\begin{equation}\label{eq1}
    M=U\Sigma V^{*}=\sum_{i=1}^{\min(m,n)}\sigma_{i}u_{i}v_{i}^{*}
\end{equation}
where $\{u_{i}\}_{i\in [m]}$ and $\{v_{j}\}_{j\in [n]}$ are columns of matrices $U$ and $V$ and thus the left and right singular vectors of matrix $M$, respectively, and $\sigma_{i}$ denotes the $i$-th singular value (the $i$-th entry of the diagonal matrix $\Sigma$) for each $i\in [\min(m,n)]$. 

For any matrix $M\in \Complex^{m\times n}$,
we denote the set of all singular values of $M$ as $s(M)$. We denote its $i$-th singular value (in non-increasing order) as $\sigma_{i}(M)$, i.e., the value~$\sigma_i$ in the decomposition of Equation (\ref{eq1}). We write $\sigma_{\max}(M)$ the largest singular value (i.e., $\sigma_{\max}(M)=\sigma_1(M)$), and write $\sigma_{\min}(M)$ the smallest non-zero singular value. We define the $\ell^{2}$ condition number of $M$ as $\kappa_{2}(M)= \sigma_{\max}(M)/\sigma_{\min}(M)\geq 1$. Note that with this definition, $\kappa_{2}$ is well defined even for singular matrices.

In this paper, we will use the following inequality by Weyl \cite{HWI12} quite often.
\begin{lemma}[Weyl's inequality \cite{HWI12}]\label{weyl}
    For two matrices $M\in\Complex^{m\times n}$, $N\in\Complex^{m\times n}$ and any $i\in [\min(m,n)]$, 
    $\absl*{\sigma_{i}(M)-\sigma_{i}(N)}\leq \normL{M-N}.$
\end{lemma}\vspace{2mm}

\noindent{\bf Singular Value Transformation.}
We are now ready to introduce the Singular Value Transformation.

\begin{definition}[Singular Value Transformation]\label{def:SVT}
    For any function $f\colon \Realnn\to \Realnn$ such that $f(0)=0$, the Singular Value Transformation associated to $f$ is the function denoted $\Phi_f$ that maps any matrix $M\in \Complex^{m\times n}$ to the matrix $\Phi_{f}(M)\in\Complex^{m\times n}$ defined as follows: 
    \[\Phi_{f}(M)=\sum_{i=1}^{\min(m,n)}f(\sigma_{i})u_{i}v_{i}^{*},\]
    where the $\sigma_i$'s, the $u_i$'s and the $v_i$'s correspond to the SVD of $M$ given in Eq.~$(\ref{eq1})$.
\end{definition}

It is easy to check that the value $\Phi_f(M)$ does not depend on the SVD of $M$ chosen in the definition (i.e., it does not depend on which $U$ and which $V$ are chosen). Also note that from our requirement on the function $f$, the rank (i.e., the number of nonzero singular values) of $\Phi_{f}(M)$ is never larger than the rank of $M$. 

The Moore-Penrose pseudo-inverse of matrix $M$ is the matrix $M^{+} = \sum_{i=1}^{k}\sigma_{i}^{-1}v_{i}u_{i}^{*}$, where $k$ is the rank of the matrix $M$. Note that we only consider non-trivial singular values of the matrix. As in the introduction, we define the inverse function $\inv\colon\Realnn\to \Realnn$ such that $\inv(0)=0$ and $\inv(x)=1/x$ for $x>0$. Then we have $\Phi_{\inv}(M^\ast)=M^{+}$. Note that $MM^{+} = M\Phi_{\inv}(M^{*}) = \Pi_{\col(M)}$ and $M^{+}M = \Phi_{\inv}(M^{*}) M= \Pi_{\row(M)}$, where $\Pi_{\col(M)}$ denotes the orthogonal projector into the column space of $M$ and $\Pi_{\row(M)}$ denotes the orthogonal projector into the row space of $M$.

\subsection{\texorpdfstring{\boldmath {$\ell^2$}}{L2}-norm sampling}
We now present the assumptions to sample from a matrix and then introduce the technique of $\ell^2$-norm sampling that has been used in previous works \cite{ADL19,CHI18,FKV04,GLT18,TAN19}.\vspace{2mm}

\noindent{\bf Sample accesses to matrices.} 
Let $M\in \Complex^{m\times n}$ be a matrix. We say that we have sample access to $M$ if the following conditions hold:
\begin{enumerate}
        \item We can sample from the probability distribution $\mathcal{R}_M\colon [m]\to[0,1]$ defined as
        $\mathcal{R}_M(i)=\frac{\nor{M_{(i,.)}}^{2}}{\normF{M}^{2}}$ for any $i\in[m]$.
        \item For each $i\in[m]$, 
        we can sample from the probability distribution $\mathcal{R}_M^i\colon [n]\to[0,1]$ 
        defined as
        $\mathcal{R}_M^i(j)=\frac{|M_{(i,j)}|^2}{\nor{M_{(i,.)}}^{2}}$ for any $j\in[n]$. (Note that $\mathcal{R}_M^i$ is precisely the distribution $\mathcal{P}_{u}$ introduced in Section \ref{sec-pre}, where $u$ is the $i$-th row of $M$.) 
\end{enumerate}

We define sample access to a vector $v\in\Complex^m$ using the same definition, by taking the matrix $M\in\Complex^{m\times 1}$ that has $v$ as unique row. Note that with this definition, the distribution $\mathcal{R}_M$ is precisely the distribution $\mathcal{P}_v$ introduced in Section \ref{subsec-not}.

For an algorithm handling matrices and vectors using sample accesses, the sample complexity of the algorithm is defined as the total number of samples used by the algorithm. \vspace{2mm}

\noindent{\bf \boldmath {$\ell^2$}-norm sampling.}
Let $M\in \Complex^{m\times n}$ be a matrix for which we have sample access. Consider the following process. For some integer $q\ge 1$, sample $q$ row indices $p_1,p_2,\ldots,p_q\in[m]$ using the probability distribution $\mathcal{R}_M$ and then form the matrix $N\in \Complex^{q\times n}$ by defining 
\[N_{(i,.)}=\frac{M_{(p_{i},.)}}{\nor{M_{(p_{i},.)}}}\frac{\normF{M}}{\sqrt{q}}\]
for each $i\in[q]$. Note that this corresponds to selecting the rows with indices $p_1,\ldots,p_q$ of $M$ and re-normalizing them. We will also use the following fact which is easy to observe using the definition of matrix $N$ :
\begin{equation}
    \normF{N}=\normF{M}
    \label{eq:clos}
\end{equation}

The central insight of the $\ell^2$-norm sampling approach introduced in \cite{FKV04} is that the matrix $N$ obtained by this process is in some sense close enough to $M$ to be able to perform several interesting calculations. We will in particular use the following result
that shows that when~$q$ is large enough, with high probability the matrix $N^{*}N$ is close to the matrix $M^{*}M$.

\begin{lemma}[Lemma 2 in \cite{FKV04}]\label{msam}
    For any $\eta \in (0,1)$, any $\beta>0$ and for $q\geq \frac{1}{\eta \beta^{2}}$, the inequality $\normF{M^{*}M-N^{*}N}\leq \beta\normF{M}^{2}$ holds with probability at least $1-\eta$.
\end{lemma}



\subsection{Data structures for storing matrices}\label{sub:data}
The following proposition shows that there exist low over-head data structures that enable sampling access to matrices.

\begin{proposition}[\cite{TAN19}]\label{tanm}
    There exists a tree-like data structure that stores a matrix $M\in \Complex^{m\times n}$ in $O(a\log^{2}(mn))$ space, where $a$ denotes the number of non-zero entries of $M$, and supports the following operations:
    \begin{itemize}
        \item[1)] Output $\normF{M}^{2}$ in $O(1)$ time;
        \item[2)] Read and update an entry $M_{(i,j)}$ in $O(\log^{2}{(mn)})$ time;
        \item[3)] Output $\nor{M_{(i,.)}}$ in $O(\log^{2}{(m)})$ time;
        \item[4)] Sampling from $\mathcal{R}_M$ in $O(\log^{2}{(mn)})$ time;
        \item[5)] For any $i\in[m]$, sampling from $\mathcal{R}_M^i$ in $O(\log^{2}{(mn)})$ time.
    \end{itemize}
\end{proposition}

The data structure of Proposition \ref{tanm} can naturally be used to store vectors as well. 

We will need the following two technical lemma in our main algorithms. Lemma \ref{3dot} shows that a vector-matrix-vector product can be efficiently approximated given sampling access. Lemma \ref{samp} states that, given sampling access to $k$ vectors represented by a $n\times k$ matrix, sampling from their linear combination is possible.

\begin{lemma}[\cite{CHI18}]\label{3dot}
    Let $v\in \Complex^{m}$ and $w\in \Complex^{n}$ be two vectors and $M\in\Complex^{m\times n}$ be a matrix, all stored in the data structure specified in Proposition \ref{tanm}.
    Then for any $\epsilon'>0$ and $\delta>0$, the value $v^{*}Mw$ 
    can be approximated with additive error $\epsilon'$ with probability at least $1-\delta$ in sample complexity $O\left(\frac{\nor{v}^{2}\nor{w}^{2}\normF{M}^{2}}{\epsilon'^{2}}\log{(\frac{1}{\delta})}\right)$ and time complexity $O\left(\frac{\nor{v}^{2}\nor{w}^{2}\normF{M}^{2}}{\epsilon'^{2}}\polylog{\left(\frac{mn}{\delta}\right)}\right)$.
\end{lemma}

\begin{lemma}[\cite{TAN19}]\label{samp}
    Let $M\in \Complex^{n\times k}$ be a matrix stored in the data structure specified in Proposition \ref{tanm}. Let $v\in \Complex^{k}$ be an input vector. Then a sample from $Mv$ can be obtained in expected sample complexity $O\left(k^{2}C(M,v)\right)$ and expected time complexity $O\left(k^{2}C(M,v)\log^2(nk)\right)$, where
    $C(M,v) = \frac{\sum_{i=1}^{k}\nor{v_{i}M_{(.,i)}}^{2}}{\nor{Mv}^{2}}$.
\end{lemma}

\section{Formal Versions and Proofs of the Main Theorems}\label{sec:main}
We now give the formal versions of Theorems \ref{th1} and \ref{th2} presented in the introduction. In this section, $\kappa_{2}$ will always denote the $\ell^{2}$ condition number of the matrix $A$. We define the intervals $L$ and $Q$ (which depend on $A$) as follows:\vspace{-3mm}

\begin{equation}\label{eq:intL}
    L= \left[\frac{\normL{A}}{\sqrt{2}\kappa_2}, \frac{\normL{A}}{\sqrt{2}\kappa_{2}}\sqrt{\left(2\kappa^{2}_{2}+1\right)}\right] \hspace{3mm}
    \textrm{ and }
    \hspace{3mm}
    Q = \left[\frac{\normL{A}^{2}}{2\kappa^{2}_{2}}, \frac{\normL{A}^{2}}{2\kappa^{2}_{2}}\left(2\kappa^{2}_{2}+1\right)\right].
\end{equation}

\addtocounter{theorem}{-8}
\begin{theorem}[Formal Version]
    Let $f\colon\Realnn\to\Realnn$ be any function such that $f(0)=0$. 
    For any $\eta>0$ and any sufficiently small $\epsilon_1>0$, there exists a classical algorithm that has sampling access as in Proposition \ref{tanm} to a matrix $A \in \Complex^{m\times n}$ and to a non-zero vector $b\in \Complex^{m}$, receives as input an index $i\in [n]$ and has the following behavior: if $f$ is differentiable on the set $L$, the algorithm outputs with probability at least $1-\eta$ a value $\lambda$ such that
    $|(\Phi_{f}(A^{*})b)_i-\lambda|\leq \epsilon_{1}$,
    using
    \[O\left(\frac{\normF{A}^{8}\nor{b}^{4}\kappa_{2}^{4}}{\epsilon_{1}^{4}\eta}\left(\frac{\kappa_{2}}{\normL{A}}\right)^{6}\Omega^{2}\left\{\phi+3\sqrt{2}\Omega\frac{\kappa_{2}}{\normL{A}}\right\}^{2}\polylog{\left(\frac{mn}{\eta}\right)}\right)\]
    samples and
    \[O\left(\frac{\normF{A}^{12}\normL{b}^{6}\kappa_{2}^{12}}{\epsilon_{1}^{6}\eta^{3}}\left\{\phi+7\sqrt{2}\Omega\frac{\kappa_{2}}{\normL{A}}\right\}^{6}\polylog{\left(mn\right)}\right)\]
    time complexity, where $\Omega = \max_{\sigma\in L}\absl*{f(\sigma)}$ and $\phi = \max_{\sigma\in L}\absl*{f'(\sigma)}$.
\end{theorem}

\begin{theorem}[Formal Version]
    Let $f\colon\Realnn\to\Realnn$ be any function such that $f(0)=0$ and $f(x)>0$ for all $x>0$. 
    For any $\eta>0$ and any sufficiently small $\epsilon_2>0$, there exists a classical algorithm that has sampling access as in Proposition \ref{tanm} to a matrix $A \in \Complex^{m\times n}$ and to a non-zero vector $b\in \Complex^{m}$ and has the following behavior: if $f$ is differentiable on the set $L$ and the projection of $b$ on the column space of $\Phi_{f}(A^*)$ has norm $\Omega(\nor{b})$, with probability at least $1-\eta$ the algorithm samples from a distribution which is $\epsilon_{2}$-close in total variation distance to the distribution $\mathcal{P}_{\Phi_{f}(A^{*})b}$,
    using
    \[O\left(\frac{\normF{A}^{10}\kappa_{2}^{14}}{\epsilon_{2}^{4}\eta^{2}\normL{A}^{6}}\left(\frac{\Omega}{\omega}\right)^{2}\left\{\frac{\phi}{\omega}+3\sqrt{2}\frac{\Omega}{\omega}\frac{\kappa_{2}}{\normL{A}}\right\}^{4}\polylog{\left(mn\right)}\right)\]
    samples and 
    \[O\left(\frac{\normF{A}^{12}\kappa_{2}^{12}}{\epsilon_{2}^{6}\eta^{3}}\left\{\frac{\phi}{\omega}+7\sqrt{2}\frac{\Omega}{\omega}\frac{\kappa_{2}}{\normL{A}}\right\}^{6}\polylog{\left(mn\right)}\right)\]
    time complexity,
    where $\Omega = \max_{\sigma\in L}\absl*{f(\sigma)}$, $\phi = \max_{\sigma\in L}\absl*{f'(\sigma)}$ and $\omega = \min_{\sigma\in L}\absl*{f(\sigma)}$.
\end{theorem}
\addtocounter{theorem}{+7}

Theorems \ref{th1} and \ref{th2} are stated for a fixed function $f$ and their correctness is guaranteed for matrices $A$ such that $f$ is differentiable on $L$ (remember that $L$ depends on $A$).  Another way of interpreting these theorems is as follows: for a matrix $A$ and vector $b$ (given as inputs), the algorithms of Theorems \ref{th1} and \ref{th2} work for any function $f\colon\Realnn\to\Realnn$ with $f(0)=0$ (and $f(x)>0$ $\forall$ $x>0$ for Theorem 2) that is differentiable in the set $L$.

Section \ref{sec:main} is organized as follows. Section \ref{sub1} presents a crucial lemma that gives an upper bound on $\normF{\Phi_{g}(X)-\Phi_{g}(Y)}$ in terms of $\normF{X-Y}$, the values of $g$ and the values of its derivative $g'$. In Section \ref{subsec:alg1} we present our central procedure, which performs row and column sampling to compute a matrix $P'\in\Complex^{r\times c}$, and analyze this procedure using the lemma proved in Section \ref{sub1}. Finally, in Section \ref{sub3} we prove Theorems \ref{th1} and \ref{th2} by applying appropriate post-processing to the matrix $P'$. 

\subsection{Bound on the distance between two singular value transformations}\label{sub1}
\vspace{-2mm}
The following lemma uses a result from \cite{MIC14} in order to derive an upper bound on the distance between two singular value transformations of positive semi-definite matrices. The proof can be found in Appendix \ref{appA}.\vspace{-1mm}

\begin{lemma}\label{fpsd}
    Let $X, Y \in\Complex^{m\times m}$ be two $m\times m$ positive semi-definite matrices, and write $S=\Conv\left((s(X)\cup s(Y))\setminus \{0\}\right)$. For any function $g\colon\Realnn\to\Realnn$ such that $g(0)=0$ and $g$ is differentiable in $S$, we have:\vspace{-3mm}
    $\normF{\Phi_{g}(X)-\Phi_{g}(Y)}\leq \normF{X-Y}\cdot \max_{\sigma\in S}\Big\{\absl*{g'(\sigma)}+\absl*{\frac{g(\sigma)}{\sigma}}\Big\}$.
\end{lemma}\vspace{-3mm}

\subsection{Core procedure}\label{subsec:alg1}
Let us consider Algorithm \ref{alg1} 
below. The goal of this subsection is to analyze its behavior.

\begin{algorithm}[ht]
\textbf{Parameters:} Three real numbers $\theta, \gamma \in \left(0,\frac{\normL{A}^{2}}{4\kappa^{2}_{2}\normF{A}^{2}}\right)$ and $\eta\in (0,1)$\\
\textbf{Input:} $A\in \Complex^{m\times n}$ 
stored in the data structure specified in Proposition \ref{tanm}
    \begin{algorithmic}[1]
        \State Set $r = 
        \ceil{3/(\eta\theta^{2})}$.
        \State Set $c = 
        \ceil{3/(\eta\gamma^{2})}$.
        \State Sample $r$ row indices $p_{1}$,...., $p_{r}$ using operation 4) of Proposition \ref{tanm}. Let $S\in \Complex^{r\times n}$ be the matrix whose $s$-th row is $S_{(s,.)} = \frac{A_{(p_{s},.)}}{\nor{A_{(p_{s},.)}}}\frac{\normF{A}}{\sqrt{r}}$, for each $s\in[r]$.
        \State Sample $c$ column indices $q_{1}$,...., $q_{c}$ by repeating the following procedure $c$ times: sample a row index $s\in[r]$ uniformly at random and then sample a column index $q\in[n]$ with probability $\frac{|S_{(s,q)}|^{2}}{\nor{S_{(s,.)}}^{2}}=\frac{|A_{(p_{s},q)}|^{2}}{\nor{A_{(p_{s},.)}}^{2}}$ using operation 5) of Proposition \ref{tanm}. 
        \State 
        Define the matrix $W\in \Complex^{r\times c}$ such that $W_{(s,t)} = \frac{S_{(s,q_{t})}}{\nor{S_{(.,q_{t})}}}\frac{\normF{S}}{\sqrt{c}} = \frac{S_{(s,q_{t})}}{\nor{S_{(.,q_{t})}}}\frac{\normF{A}}{\sqrt{c}}$, for each $(s,t)\in [r]\times[c]$.
        Query all the entries of $A$ corresponding to entries of $W$ using operation~2) of Proposition \ref{tanm}.
        \State Compute the singular value decomposition of matrix $W$.
        \State Compute the matrix $P'=\Phi_{\inv}(W)\Phi_{f}(W^{*})\Phi_{\inv}(W)\Phi_{\inv}(W^{*})$ using the output of the SVD step.
\end{algorithmic}
\caption{Computing the matrix $P'$.}
\label{alg1}
\end{algorithm}

The sampling process of Steps 3--5 is exactly the same as in prior works \cite{ADL19,CHI18,FKV04,GLT18,TAN19}, but with different values for $c$ and $r$. The following lemma analyzes the matrices $S$ and $W$ obtained by this process.
The proof, which is the same as in these prior works (but with different values for $c$ and $r$), can be found in Appendix \ref{appB}.
\begin{lemma}\label{ineq}
    For any input matrix $A$ and any parameters $(\theta, \gamma, \eta)$ in the specified range, with probability at least $1-2\eta/3$ the following statements are simultaneously true for the matrices $S$ and $W$ computed by Algorithm \ref{alg1}:
    \begin{align}
    &\normF{S}=\normF{A}\label{eq:cond1}\\
    &\normF{A^{*}A-S^{*}S}\leq \theta\normF{A}^{2},\label{eq:cond2}\\ 
    &\normF{SS^{*}-WW^{*}}\leq \gamma\normF{S}^{2},\label{eq:cond3}\\
    &\sigma_{\min}(S) > 
    \frac{\normL{A}}{\sqrt{2}\kappa_2},\hspace{8mm}
        \sigma_{\max}(S) < 
        \frac{\normL{A}}{\sqrt{2}\kappa_{2}}\sqrt{\left(2\kappa^{2}_{2}+1\right)},\label{eq:cond4}\\
    &\sigma_{\min}(W) > \frac{\normL{A}}{\sqrt{2}\kappa_2},\hspace{5mm} 
        \sigma_{\max}(W) < \frac{\normL{A}}{\sqrt{2}\kappa_{2}}\sqrt{\left(2\kappa^{2}_{2}+1\right)}\label{eq:cond5}.
    \end{align}
\end{lemma}

Lemma \ref{ineq} above guarantees in particular that with high probability all the nontrivial singular values of the matrix $S$ and $W$ are in the interval $L$ defined in Equation (\ref{eq:intL}).

The main originality of our approach is Step 7 of Algorithm \ref{alg1}, which we now analyze. Let us define the matrix 
$
P=\Phi_{\inv}(S)\Phi_{f}(S^{*})\Phi_{\inv}(S)\Phi_{\inv}(S^{*}).
$
The following lemma shows that the output $P'$ of Algorithm \ref{alg1} is close to the matrix~$P$. Due to space constraints, here we only give a sketch of the proof. A complete proof can be found in Appendix \ref{appB}.

\begin{lemma}\label{ppcl}
    Assume that Statements (\ref{eq:cond1})-(\ref{eq:cond5}) of Lemma \ref{ineq} all hold (which happens with probability at least $1-2\eta/3$). Assume that $f$ is differentiable in $L$ and $f(0)=0$. Then the matrix $P'\in \Complex^{r \times r}$ obtained as the output of Algorithm \ref{alg1} satisfies the following inequality, where $\Omega = \max_{\sigma\in L}\absl*{f(\sigma)}$ and $\phi = \max_{\sigma\in L}\absl*{f'(\sigma)}$.
    \begin{equation}\label{eq:ppcl}
        \normF{P'-P}\leq 2\gamma\normF{A}^{2}\left(\frac{\kappa_{2}}{\normL{A}}\right)^{4}\left\{\phi+7\sqrt{2}\Omega\frac{\kappa_{2}}{\normL{A}}\right\},
    \end{equation}
\end{lemma}

\begin{proof}[Sketch of the proof]
    Let us define a function $h\colon \Realnn\to\Realnn$ as follows. For any $\sigma\in Q$ we define 
    $h(\sigma)=f(\sqrt{\sigma})\inv(\sqrt{\sigma})=f(\sqrt{\sigma})/\sqrt{\sigma}$,
    we define $h(0)=f(0)\inv(0)=0$, and we define $h(\sigma)$ arbitrarily when $\sigma\notin Q\cup\{0\}$. Since $f$ is differentiable in $L$, the function $h$ is differentiable in $Q$. From Equations (\ref{eq:cond4}) and (\ref{eq:cond5}) we know that
    $\Conv\left(s(SS^{*})\cup s(WW^{*})\setminus\{0\}\right) \subset Q$
    and can write $\Phi_{h}(SS^{*})=\Phi_{\inv}(S)\Phi_{f}(S^{*})$ and $\Phi_{h}(WW^{*})=\Phi_{\inv}(W)\Phi_{f}(W^{*})$. 
    
    Using the definition of $P$ and $P'$, we now have
    \[\begin{split}
        \normF{P'-P}\\
        &\hspace{-17mm}=\normF{\Phi_{\inv}(W)\Phi_{f}(W^{*})\Phi_{\inv}(W)\Phi_{\inv}(W^{*})-\Phi_{\inv}(S)\Phi_{f}(S^{*})\Phi_{\inv}(S)\Phi_{\inv}(S^{*})}\\
        &\hspace{-17mm}= \normF{\Phi_{h}(WW^{*})\Phi_{\inv}(WW^{*})-\Phi_{h}(SS^{*})\Phi_{\inv}(SS^{*})}\\
        &\hspace{-17mm}=\normF{\left\{\Phi_{h}(WW^{*})-\Phi_{h}(SS^{*})\right\}\Phi_{\inv}(WW^{*})+\Phi_{h}(SS^{*})\left\{\Phi_{\inv}(WW^{*})-\Phi_{\inv}(SS^{*})\right\}}\\
        &\hspace{-17mm}\leq \normF{\left\{\Phi_{h}(WW^{*})-\Phi_{h}(SS^{*})\right\}\Phi_{\inv}(WW^{*})}+\normF{\Phi_{h}(SS^{*})\left\{\Phi_{\inv}(WW^{*})-\Phi_{\inv}(SS^{*})\right\}}\\
        &\hspace{-17mm}\leq \normL{\Phi_{\inv}(WW^{*})}\normF{\Phi_{h}(WW^{*})-\Phi_{h}(SS^{*})}+\normL{\Phi_{h}(SS^{*})}\normF{\Phi_{\inv}(WW^{*})-\Phi_{\inv}(SS^{*})}.
    \end{split}\]
      
    Using Lemma \ref{fpsd} twice for $\Phi_{h}$ and $\Phi_{\inv}$, we obtain
    \[\begin{split}
        \normF{P'-P}&\leq \normL{\Phi_{\inv}(WW^{*})}\normF{WW^{*}-SS^{*}}\left(\max_{\sigma\in Q}\left\{\absl*{h'(\sigma)}+\absl*{\frac{h(\sigma)}{\sigma}}\right\}\right)\\
        &+\normL{\Phi_{h}(SS^{*})}\normF{WW^{*}-SS^{*}}\left(\max_{\sigma\in Q}\left\{\absl*{\inv'(\sigma)}+\absl*{\frac{\inv(\sigma)}{\sigma}}\right\}\right).\\
    \end{split}\]
    
    Now we use condition (\ref{eq:cond3}). Also, since the non-trivial singular values of $SS^{*}$ and $WW^{*}$ lie in the set $Q$, the non-trivial singular values of $S$ and $W$ lie in set $L$ (i.e., if $\sigma\in Q$ then $\sigma^{1/2}\in L$). Using this observation, we can then derive the claimed upper bound by routine calculations (omitted here).
\end{proof}

The next proposition is the main result of this subsection. Due to space constrain, the proof has been deferred to Appendix \ref{appB}.
\begin{proposition}\label{ovec}
    Let $b\in \Complex^{m}$ be any non-zero vector and $\epsilon$ be any positive number such that 
    \begin{align}
        \epsilon < \frac{1}{2}\normL{A}\nor{b}\left\{\phi+3\sqrt{2}\Omega\frac{\kappa_{2}}{\normL{A}}\right\}.
        \label{eq:eps}
    \end{align} 
    Let us fix the parameters of Algorithm \ref{alg1} as follows:
    \begin{align}
        \theta&=\epsilon\left(2\normF{A}^{2}\frac{\kappa_{2}^{2}}{\normL{A}}\left\{\phi+3\sqrt{2}\Omega\frac{\kappa_{2}}{\normL{A}}\right\}\nor{b}\right)^{-1},\label{ovec1}\\
        \gamma&=\epsilon\left(2\normF{A}^{2}\frac{\kappa_{2}^{2}}{\normL{A}}\left\{\phi+7\sqrt{2}\Omega\frac{\kappa_{2}}{\normL{A}}\right\}\nor{b}\right)^{-1}.\label{ovec2}
    \end{align}
    Then, under the assumptions of Lemma \ref{ppcl}, the two vectors $x=S^{*}P'SA^{*}b$ and $\Phi_{f}(A^{*})b$ satisfy the inequality
    $\nor{x-\Phi_{f}(A^{*})b}\leq \epsilon.
    $
\end{proposition}

\subsection{Post-processing and proofs of Theorems \ref{th1} and \ref{th2}}\label{sub3}

\begin{proof}[Proof of Theorem \ref{th1}]
    Let us write 
    \begin{equation}\label{eq:ep1}
        \epsilon'=\frac{\epsilon_{1}}{4\Omega\sqrt{r\left(2\kappa_{2}^{2}+1\right)}}\left(\frac{\normL{A}}{\kappa_{2}}\right)^{2}
        \hspace{3mm}
        \textrm{ and }
        \hspace{3mm}
        \delta'=\eta/3r.
    \end{equation}
    
    The algorithm we consider for estimating the value $(\Phi_{f}(A^{*})b)_{i}$ is described below.
    
    \begin{algorithm}[ht]
    \begin{algorithmic}[1]
        \State Apply Algorithm \ref{alg1} with matrix $A$ as input, using the values $\theta$ and $\gamma$ given by Equations~(\ref{ovec1}) and (\ref{ovec2}) with $\epsilon=\epsilon_{1}/2$, and using the desired $\eta$ as parameters. This returns a matrix $P'$ and a description of a matrix $S$.
        \State Compute an estimation $z$ of the vector $SA^{*}b\in\Complex^{r\times 1}$ by estimating, for each $j\in [r]$, the quantity $S_{(j,.)}A^{*}b$ using Lemma \ref{3dot} with parameters $\epsilon'$ and $\delta'$ given by Equation (\ref{eq:ep1}).
        \State Compute the row vector $S^{*}_{(i,.)}\in\Complex^{1\times r}$ by querying all the elements in the $i$-th row of $S^{*}$ (i.e., the $i$-th column of $S$).
        \State Output the complex number $S^{*}_{(i,.)}P'z$.
    \end{algorithmic}
    \caption{Estimating $(\Phi_{f}(A^{*})b)_i$}
    \label{alg2}
    \end{algorithm}
    
    We now analyze Algorithm \ref{alg2}. Let us write $x'=S^{*}P'z\in \Complex^{n\times 1}$, where $P'$ and $z$ are the matrices and the vector computed at Steps 1 and 2 of the algorithm, respectively. Remember that $P'=\Phi_{\inv}(W)\Phi_{f}(W^{*})\Phi_{\inv}(W)\Phi_{\inv}(W^{*})$, where $W$ is the matrix computed in Algorithm~\ref{alg1}. Note that the output of Algorithm \ref{alg2} is the $i$-th coordinate of the vector $x'$.   
    
    Let us write $x=S^{*}P'SA^{*}b$. 
    From the analysis of Section \ref{subsec:alg1}, and especially Lemma \ref{ineq} and Proposition \ref{ovec}, we know that Statements (\ref{eq:cond4}) and (\ref{eq:cond5}) and the inequality $\nor{x-\Phi_{f}(A^{*})b}\leq \frac{\epsilon_{1}}{2}$ simultaneously hold with probability $1-2\eta/3$. 
    
    The vector $x'$ then satisfies the inequality
    \[\begin{split}
        \nor{x'-x}& \leq \nor{S^{*}P'z-S^{*}P'SA^{*}b}\\
        & \leq \normL{S^{*}}\normL{P'}\nor{z-SA^{*}b}\\
        & \leq \normL{S^{*}}\normL{\Phi_{\inv}(W)}\normL{\Phi_{f}(W^{*})}\normL{\Phi_{\inv}(W)}\normL{\Phi_{\inv}(W^{*})}\nor{z-SA^{*}b}\\
        &\leq \left\{\frac{\normL{A}}{\sqrt{2}\kappa_{2}}\left(2\kappa_{2}^{2}+1\right)^{1/2}\right\}\left\{\Omega\left(\frac{\sqrt{2}\kappa_{2}}{\normL{A}}\right)^{3}\right\}\nor{z-SA^{*}b},
    \end{split}\]
    where we used Statements (\ref{eq:cond4}) and (\ref{eq:cond5}) and the bound $\normL{\Phi_{f}(W^{*})}\leq \Omega$ to derive the last inequality.
    
    Lemma \ref{3dot} now guarantees that with probability at least $1-\eta/3$ we have $\nor{z-SA^{*}b}\leq \epsilon'\sqrt{r}$, which implies:
    \[\begin{split}    
        \nor{x'-x}&\leq
        \left\{\frac{\normL{A}}{\sqrt{2}\kappa_{2}}\left(2\kappa_{2}^{2}+1\right)^{1/2}\right\}\left\{\Omega\left(\frac{\sqrt{2}\kappa_{2}}{\normL{A}}\right)^{3}\right\}\left\{\frac{\epsilon_{1}}{4\Omega\sqrt{\left(2\kappa_{2}^{2}+1\right)}}\left(\frac{\normL{A}}{\kappa_{2}}\right)^{2}\right\}=\frac{\epsilon_{1}}{2}.
    \end{split}\]

    In conclusion, the inequality 
    \begin{equation}\label{xdas}
        \nor{x'-\Phi_{f}(A^{*})b}\leq \nor{x'-x}+\nor{x-\Phi_{f}(A^{*})b}\leq \epsilon_{1}
    \end{equation}
    holds with overall probability at least $1-\eta$ for sufficiently small $\epsilon_{1}>0$ (a precise upper bound can be derived by using Proposition \ref{ovec} with $\epsilon=\epsilon_{1}/2$).
    
    This implies that Algorithm~\ref{alg2} outputs, with probability at least $1-\eta$, the $i$-th coordinate of a vector $x'$ that satisfies Equation (\ref{xdas}). This proves the correctness of Algorithm $\ref{alg2}$.
    
    Let us now analyze the complexity of Algorithm \ref{alg2}. Algorithm \ref{alg1} (and thus Step 1 of Algorithm~\ref{alg2}) has time complexity dominated by the computation of the SVD of the matrix $W$, i.e.,
    \[O\left(\max\left\{r^{2}c, rc^{2}\right\}\polylog{\left(mn\right)}\right)=
    O\left(\frac{\normF{A}^{12}\normL{b}^{6}\kappa_{2}^{12}}{\epsilon_{1}^{6}\eta^{3}}\left\{\phi+7\sqrt{2}\Omega\frac{\kappa_{2}}{\normL{A}}\right\}^{6}\polylog{\left(mn\right)}\right).\]
    Algorithm~\ref{alg1} uses $r+c$ samples.
    
    Observe that $\nor{S_{(j,.)}}=\frac{\normF{A}}{\sqrt{r}}$ for any $j\in[r]$ (see Step 3 of Algorithm \ref{alg1}). 
    Step 2 of Algorithm \ref{alg2} thus uses 
    \begin{align*}
        O\left(\frac{\nor{S_{(j,.)}}^{2}\nor{b}^{2}\normF{A^{*}}^{2}}{\epsilon'^{2}}\polylog{\left(\frac{mn}{\delta}\right)}r\right)
        &=\\
        &\hspace{-32mm}
        O\left(\frac{\normF{A}^{8}\nor{b}^{4}\kappa_{2}^{4}}{\epsilon_{1}^{4}\eta}\left(\frac{\kappa_{2}}{\normL{A}}\right)^{6}\Omega^{2}\left\{\phi+3\sqrt{2}\Omega\frac{\kappa_{2}}{\normL{A}}\right\}^{2}\polylog{\left(\frac{mn}{\eta}\right)}\right)
    \end{align*}
    samples, and has the same time complexity.
    
    Finally, Step 3 of Algorithm \ref{alg2} has time complexity $O(r)$, while Step~4 has time complexity $O(r^{2})$. These two steps do not use any sample.
    
    In conclusion, the time complexity of Algorithm~\ref{alg2} is dominated by Step 1, while the sample complexity is dominated by Step 2.
\end{proof}

\begin{proof}[Proof sketch of Theorem \ref{th2}]
    Let us write 
    \begin{equation}\label{eq:ep2}
        \epsilon''=\frac{\epsilon_{2}\omega\alpha\nor{b}}{8\Omega\sqrt{r\left(2\kappa_{2}^{2}+1\right)}}\left(\frac{\normL{A}}{\kappa_{2}}\right)^{2}
                \hspace{3mm}
        \textrm{ and }
        \hspace{3mm}
        \delta''=\eta/3r,
    \end{equation}
    where $\alpha$ is a constant such that the norm of the projection of $b$ on the column space of $\Phi_{f}(A)$ is at least $\alpha\nor{b}$.
    The algorithm we use to sample from a distribution $\epsilon_{2}$-close to $\mathcal{P}_{\Phi_{f}(A^{*})b}$ is described below. 
    \begin{algorithm}[ht]
    \begin{algorithmic}[1]
    \State Apply Algorithm \ref{alg1} with matrix $A$ as input, using the values $\theta$ and $\gamma$ given by Equations~(\ref{ovec1}) and (\ref{ovec2}) with $\epsilon=\frac{\epsilon_{2}\omega\alpha}{4}\nor{b}$, and using the desired $\eta$ as parameters.
    This returns a matrix $P'$ and a description of a matrix $S$.
    \State Compute an estimation $z$ of the vector $SA^{*}b\in\Complex^{r\times 1}$ by estimating, for each $j\in [r]$, the quantity $S_{(j,.)}A^{*}b$ using Lemma \ref{3dot} with parameters $\epsilon''$ and $\delta''$ given by Equation (\ref{eq:ep2}).
    \State Compute the vector $P'z$.
    \State Use Lemma \ref{samp} to output a sample from $x'=S^{*}P'z$. 
    \end{algorithmic}
    \caption{Sample access to a distribution $\epsilon_{2}$-close to $\mathcal{P}_{\Phi_{f}(A^{*})b}$}
    \label{alg3}
    \end{algorithm} 
    
    Note that Algorithm \ref{alg3} is very similar to Algorithm \ref{alg2}: the main modification is Step 4. 
    Also note that we can use Lemma \ref{samp} since we have sample access to the columns of $S^{*}$, from the information obtained at Step 1, and we can compute the vector $P'z$ from the information obtained at Steps 1 and 2. 
    
    The complete analyses of the correctness and the complexity of Algorithm \ref{alg3}, which are similar to the analyses done for Algorithm \ref{alg2} in the proof of Theorem \ref{th2}, can be found in Appendix \ref{appC}.   
\end{proof}




\appendix
\section{Proof of Lemma \ref{fpsd}}\label{appA}
In this appendix we give the proof of Lemma \ref{fpsd}. Our proof relies on a prior work that established similar bounds for eigenvalue transformations. We first present this result in Part~\ref{sub:eig}, and then present the proof of Lemma \ref{fpsd} in Part \ref{sub:prooffpsf}.
\subsection{Eigenvalue transformations}\label{sub:eig}
Let us introduce below another transformation applicable to a diagonalizable matrix $M\in\Complex^{m\times m}$, i.e., a matrix than can be written as 
\begin{equation}\label{eq2}
    M = Q \:\diag(\lambda_1,\ldots,\lambda_m)\:Q^{-1}
\end{equation}
for some invertible matrix $Q\in\Complex^{m\times m}$ where $\diag(\lambda_1,\ldots,\lambda_m)$ denotes the $m\times m$ diagonal matrix with diagonal entries as $m$ complex numbers $\lambda_1,\ldots,\lambda_m$.  We write $e(M)=\{\lambda_1,\ldots,\lambda_m\}$, which is the set of eigenvalues of $M$. 

\begin{definition}[Eigenvalue Transformation]
For any function $f\colon \Complex\to \Complex$ such that $f(0)=0$, the Eigenvalue Transformation associated to $f$ is the function denoted $\Psi_f$ that maps any diagonalizable matrix $M\in \Complex^{m\times m}$ to the matrix $\Phi_{f}(M)\in\Complex^{m\times m}$ defined as follows: 
\[\Psi_{f}(M) = Q \:\diag(f(\lambda_1),\ldots,f(\lambda_m))\:Q^{-1},\]
where $Q$ and $\lambda_1,\ldots,\lambda_m$ correspond to the decomposition of $M$ given in Eq.~$(\ref{eq2})$.
\end{definition}

Similarly to Definition \ref{def:SVT}, due to our assumption on $f$ the eigenvalue transformation function does not increase the rank of the input matrix.

We will later use the following upper bound on the norm of $\normF{\Psi_{f}(M)-\Psi_{f}(M')}$  for diagonalizable matrices $M$ and $M'$ from \cite{MIC14}.

\begin{lemma}[Corollary 2.3 in \cite{MIC14}]\label{diag}
    Let $M$ and $M'$ be $m\times m$ diagonalizable matrices with decompositions
    \begin{align*}   
    M &= \:Q \:\diag(\lambda_1,\ldots,\lambda_m)\:Q^{-1},\\
    M' &= \:Q' \:\diag(\lambda'_1,\ldots,\lambda'_m)\:Q'^{-1}.
    \end{align*}
    For any function $f\colon\Complex\to\Complex$ we have 
    \[\normF{\Psi_{f}(M)-\Psi_{f}(M')}\leq \kappa_{2}(Q)\kappa_{2}(Q')\normF{M-M'}\cdot\max_{j\in [m],k\in [m]}\left\{\absl*{\frac{f(\lambda_{j})-f(\lambda'_{k})}{\lambda_{j}-\lambda'_{k}}}\right\},\]
    where the convention $\absl*{\frac{f(\lambda_{j})-f(\lambda'_{k})}{\lambda_{j}-\lambda'_{k}}}=0$ if $\lambda_{j}=\lambda'_{k}$ is used.
\end{lemma} 
\subsection{Proof of Lemma \ref{fpsd}}\label{sub:prooffpsf}
\begin{proof}[Proof of Lemma \ref{fpsd}]
    For a positive semi-definite matrix the singular values are equal to the eigenvalues and the matrix $Q$ in the decomposition of Equation (\ref{eq2}) can be taken as a unitary matrix. For means that for a positive semi-definite matrix, its singular value transformation is equal to its eigenvalue transformation. Note that if $Q$ is unitary then $\kappa_2(Q)=1$.
    
    Using Lemma \ref{diag} we thus obtain:
    \[
    \normF{\Phi_{g}(X)-\Phi_{g}(Y)}=
    \normF{\Psi_{g}(X)-\Psi_{g}(Y)}\leq 
    \normF{X-Y}\cdot \max_{j\in [m],k\in [m]}\left\{\absl*{\frac{g(\sigma_{j})-g(\sigma'_{k})}{\sigma_{j}-\sigma'_{k}}}\right\},
    \]
    where we write $s(X)=\{\sigma_1,\sigma_2,\ldots,\sigma_m\}$ and $s(Y)=\{\sigma'_1,\sigma'_2,\ldots,\sigma'_m\}$. 
    
    For conciseness, let us write $\delta_{jk}=\absl*{\frac{g(\sigma_{j})-g(\sigma'_{k})}{\sigma_{j}-\sigma'_{k}}}$ for any $(j,k)\in[m]\times [m]$. There are three cases:
    \begin{enumerate}
        \item For any $(j,k)$ such that $\sigma_{j}\neq 0$ and $\sigma'_{k}\neq 0$ we have $\delta_{jk}\le \max_{\sigma\in S}\absl*{g'(\sigma)}$. This happens because~$g$ is differentiable in $S$. Indeed, if we choose values $a\in S$ and $b\in S$ such that $a<b$, we can always find a value $\sigma\in [a,b]$ such that $g'(\sigma)=\frac{g(b)-g(a)}{b-a}$ by then Intermediate Value Theorem. Since this happens for all values of $a,b$, we obtain $\delta_{jk}\leq \max_{\sigma\in S}\absl*{g'(\sigma)}$.
        \item For any $(j,k)$ such that $\sigma_{j}=0$ and $\sigma'_{k}\neq 0$, or $\sigma_{j}\neq 0$ and $\sigma'_{k}=0$, we have $\delta_{jk}\le \max_{\sigma\in S}\absl*{\frac{g(\sigma)}{\sigma}}$;
        \item For any $(j,k)$ such that $\sigma_{j}=0$ and $\sigma'_{k}=0$ we have $\delta_{jk}=0$ (by convention in Lemma \ref{diag}).
    \end{enumerate}
    
    Then
    \[\max_{j\in [m],k\in [m]}\left\{\delta_{jk}\right\}\leq \max_{\sigma\in S}\left\{\absl*{g'(\sigma)}, \absl*{\frac{g(\sigma)}{\sigma}}\right\}.\]
    
    Therefore,
    \[\normF{\Phi_{g}(X)-\Phi_{g}(Y)}\leq \normF{X-Y}\cdot \max_{\sigma\in S}\left\{\absl*{g'(\sigma)}+\absl*{\frac{g(\sigma)}{\sigma}}\right\},\]
    as claimed.
\end{proof}

\section{Proofs of Lemma \ref{ineq}, Lemma \ref{ppcl} and Proposition \ref{ovec}}\label{appB}
\begin{proof}[Proof of Lemma \ref{ineq}]
    Equation (\ref{eq:clos}) suggests Statement (\ref{eq:cond1}) always holds. Hence $\normF{S}=\normF{A}$.
    
    Using Lemma \ref{msam} twice, the following two inequalities simultaneously hold for matrices $A$, $S$ and~$W$ in Algorithm \ref{alg1} with probability at least $1-2\eta/3$:
    \[\normF{A^{*}A-S^{*}S}\leq \theta\normF{A}^{2},\]
    \[\normF{SS^{*}-WW^{*}}\leq \gamma\normF{S}^{2}.\]
    
    Thus with probability at least $1-2\eta/3$, Statements (\ref{eq:cond1}), (\ref{eq:cond2}) and (\ref{eq:cond3}) simultaneously hold. We now show that in this case, Statements (\ref{eq:cond4}) and (\ref{eq:cond5}) always hold. 
    
    Using Weyl's inequality (Lemma \ref{weyl}) for $k=\min\{\rank(A^{*}A),\rank(S^{*}S)\}$, and the above conditions, we have
    \[\absl*{\sigma_{k}(S^{*}S)-\sigma_{k}{(A^{*}A)}}\leq \normL{A^{*}A-S^{*}S}\leq \normF{A^{*}A-S^{*}S}\leq \theta\normF{A}^{2} < \frac{\normL{A}^{2}}{4\kappa^{2}_{2}}.\]
    
    Now $\rank(S^{*}S)\leq \rank(A^{*}A)$ (as we sample rows from $A$). Since $\sigma_{\min}{(A^{*}A)}=\frac{\normL{A}^{2}}{\kappa^{2}_{2}}$ (by the definition of $\kappa_{2}$), we get 
    \[\sigma^{2}_{\min}(S)=\sigma_{\min}(S^{*}S) > \sigma_{k}(A^{*}A)-\frac{\normL{A}^{2}}{4\kappa^{2}_{2}} > \sigma_{\min}(A^{*}A)-\frac{\normL{A}^{2}}{4\kappa^{2}_{2}} = \frac{3\normL{A}^{2}}{4\kappa^{2}_{2}} >  \frac{\normL{A}^{2}}{2\kappa^{2}_{2}}.\] 
    
    By a similar argument we get 
    \[\sigma^{2}_{\max}(S)=\sigma_{\max}(S^{*}S) < \sigma_{\max}{(A^{*}A)} + \frac{\normL{A}^{2}}{4\kappa^{2}_{2}}=\normL{A}^{2}\left(1+\frac{1}{4\kappa^{2}_{2}}\right)<\frac{\normL{A}^{2}}{2\kappa_2^2}\left(2\kappa^{2}_{2}+1\right).\]
    
    Using Weyl's inequality (Lemma \ref{weyl}) again for $k'=\min\{\rank(SS^{*}),\rank(WW^{*})\}$ we obtain:
    \[\begin{split}
        \absl*{\sigma_{k'}(WW^{*})-\sigma_{k'}{(SS^{*})}}&\leq \normF{SS^{*}-WW^{*}}\\
        &\leq \gamma\normF{S}^{2}\\
        &\leq \frac{\normL{A}^{2}}{4\kappa_{2}^{2}\normF{A}^{2}} \normF{A}^{2}<\frac{\normL{A}^{2}}{4\kappa_{2}^{2}}.
    \end{split}\]
    
    Since $\sigma_{\min}(SS^{*})>\frac{3\normL{A}^{2}}{4\kappa^{2}_{2}}$ we finally obtain the lower bound \[\sigma_{\min}^{2}(W)=\sigma_{\min}(WW^{*}) > \frac{\normL{A}^{2}}{2\kappa^{2}_{2}}.\] 
    
    A similar argument gives the upper bound $\sigma_{\max}^{2}(W) < \normL{A}^{2}\left(1+\frac{1}{2\kappa^{2}_{2}}\right)$. 
\end{proof}

\begin{proof}[Proof of Lemma \ref{ppcl}]
    Let us define a function $h\colon \Realnn\to\Realnn$ as follows. For any $\sigma\in Q$ we define 
    $h(\sigma)=f(\sqrt{\sigma})\inv(\sqrt{\sigma})=f(\sqrt{\sigma})/\sqrt{\sigma}$,
    we define $h(0)=f(0)\inv(0)=0$, and we define $h(\sigma)$ arbitrarily when $\sigma\notin Q\cup\{0\}$. Since $f$ is differentiable in $L$, the function $h$ is differentiable in $Q$. From Equations (\ref{eq:cond4}) and (\ref{eq:cond5}) we know that
    $\Conv\left(s(SS^{*})\cup s(WW^{*})\setminus\{0\}\right) \subset Q$
    and can write $\Phi_{h}(SS^{*})=\Phi_{\inv}(S)\Phi_{f}(S^{*})$ and $\Phi_{h}(WW^{*})=\Phi_{\inv}(W)\Phi_{f}(W^{*})$. 
    
    Using the definition of $P$ and $P'$, we now have
    \[\begin{split}
        \normF{P'-P}\\
        &\hspace{-17mm}=\normF{\Phi_{\inv}(W)\Phi_{f}(W^{*})\Phi_{\inv}(W)\Phi_{\inv}(W^{*})-\Phi_{\inv}(S)\Phi_{f}(S^{*})\Phi_{\inv}(S)\Phi_{\inv}(S^{*})}\\
        &\hspace{-17mm}= \normF{\Phi_{h}(WW^{*})\Phi_{\inv}(WW^{*})-\Phi_{h}(SS^{*})\Phi_{\inv}(SS^{*})}\\
        &\hspace{-17mm}=\normF{\left\{\Phi_{h}(WW^{*})-\Phi_{h}(SS^{*})\right\}\Phi_{\inv}(WW^{*})+\Phi_{h}(SS^{*})\left\{\Phi_{\inv}(WW^{*})-\Phi_{\inv}(SS^{*})\right\}}\\
        &\hspace{-17mm}\leq \normF{\left\{\Phi_{h}(WW^{*})-\Phi_{h}(SS^{*})\right\}\Phi_{\inv}(WW^{*})}+\normF{\Phi_{h}(SS^{*})\left\{\Phi_{\inv}(WW^{*})-\Phi_{\inv}(SS^{*})\right\}}\\
        &\hspace{-17mm}\leq \normL{\Phi_{\inv}(WW^{*})}\normF{\Phi_{h}(WW^{*})-\Phi_{h}(SS^{*})}+\normL{\Phi_{h}(SS^{*})}\normF{\Phi_{\inv}(WW^{*})-\Phi_{\inv}(SS^{*})}.
    \end{split}\]
    
    Using Lemma \ref{fpsd} twice for $\Phi_{h}$ and $\Phi_{\inv}$, we obtain
    \[\begin{split}
        \normF{P'-P}&\leq \normL{\Phi_{\inv}(WW^{*})}\normF{WW^{*}-SS^{*}}\left(\max_{\sigma\in Q}\left\{\absl*{h'(\sigma)}+\absl*{\frac{h(\sigma)}{\sigma}}\right\}\right)\\
        &+\normL{\Phi_{h}(SS^{*})}\normF{WW^{*}-SS^{*}}\left(\max_{\sigma\in Q}\left\{\absl*{\inv'(\sigma)}+\absl*{\frac{\inv(\sigma)}{\sigma}}\right\}\right).\\
    \end{split}\]
    
    Now using (\ref{eq:cond3}) we obtain
    \[\begin{split}
        \normF{P'-P}&\leq \gamma\normF{S}^{2}\normL{\Phi_{\inv}(WW^{*})}\left(\max_{\sigma\in Q}\left\{\absl*{h'(\sigma)}+\absl*{\frac{h(\sigma)}{\sigma}}\right\}\right)\\
        &+\gamma\normF{S}^{2}\normL{\Phi_{h}(SS^{*})}\left(\max_{\sigma\in Q}\left\{\absl*{\inv'(\sigma)}+\absl*{\frac{\inv(\sigma)}{\sigma}}\right\}\right).
    \end{split}\]
    
    Since the nontrivial singular values of $SS^{*}$ and $WW^{*}$ lie in the set $Q$, the nontrivial singular values of $S$ and $W$ lie in set $L$ (i.e., if $\sigma\in Q$ then $\sigma^{1/2}\in L$). We can thus write the above equation as:
    \[\begin{split}
        \normF{P'-P}&\leq
        \gamma\normF{S}^{2}\left(\max_{\sigma\in Q}\left\{\absl*{\inv(\sigma)}\right\}\right)\left(\max_{\sigma\in Q}\left\{\absl*{h'(\sigma)}+\absl*{\frac{h(\sigma)}{\sigma}}\right\}\right)\\
        &+\gamma\normF{S}^{2}\left(\max_{\sigma\in Q}\left\{\absl*{h(\sigma)}\right\}\right)\left(\max_{\sigma\in Q}\left\{\absl*{\inv'(\sigma)}+\absl*{\frac{\inv(\sigma)}{\sigma}}\right\}\right).
    \end{split}\]
    
    By routine calculation,
    \[\max_{\sigma\in Q}\left\{\absl*{h'(\sigma)}\right\} = \max_{\sigma\in Q}\left\{\absl*{\frac{\sqrt{\sigma}f'(\sqrt{\sigma})-f(\sqrt{\sigma})}{2(\sqrt{\sigma})^{3}}}\right\}\leq \max_{\sigma\in L}\left\{\absl*{\frac{f'(\sigma)}{2\sigma^{2}}}+\absl*{\frac{f(\sigma)}{2\sigma^{3}}}\right\},\]
    which implies
    \begin{align*}
        \max_{\sigma\in Q}\left\{\absl*{\inv(\sigma)}\right\}\cdot\left(\max_{\sigma\in Q}\left\{\absl*{h'(\sigma)}+\absl*{\frac{h(\sigma)}{\sigma}}\right\}\right)\\
        &\hspace{-10mm}\leq  \max_{\sigma\in L}\left\{\frac{1}{\sigma^{2}}\right\}\cdot\left(
        \max_{\sigma\in Q}\left\{\absl*{h'(\sigma)}\right\}+\max_{\sigma\in L}\left\{\absl*{\frac{f(\sigma)}{\sigma^{3}}}\right\}\right)\\
        &\hspace{-10mm}\leq  \phi\max_{\sigma\in L}\left\{\frac{1}{2\absl*{\sigma}^{4}}\right\} + \Omega\max_{\sigma\in L}\left\{\frac{3}{2\absl*{\sigma}^{5}}\right\}\\
        &\hspace{-10mm}\leq  2\left(\frac{\kappa_{2}}{\normL{A}}\right)^{4}\left\{\phi+3\sqrt{2}\Omega\frac{\kappa_{2}}{\normL{A}}\right\}.
    \end{align*}
    
    Similarly, we get    
    \begin{align*}
        \max_{\sigma\in Q}\left\{\absl*{h(\sigma)}\right\}\cdot \left(\max_{\sigma\in Q}\left\{\absl*{\inv'(\sigma)}+\absl*{\frac{\inv(\sigma)}{\sigma}}\right\}\right)&\leq  \max_{\sigma\in L}\left\{\absl*{\frac{f(\sigma)}{\sigma}}\right\}\cdot \max_{\sigma\in L}\left\{\frac{1}{\sigma^{4}}+\frac{1}{\sigma^{4}}\right\}\\
        &\leq  \Omega\max_{\sigma\in L}\left\{\frac{2}{\absl*{\sigma}^{5}}\right\}\\
        &\leq  2\left(\frac{\kappa_{2}}{\normL{A}}\right)^{4}\left\{4\sqrt{2}\Omega\frac{\kappa_{2}}{\normL{A}}\right\}.
    \end{align*}

    Using these inequalities we finally obtain the upper bound    
    \begin{align*}
        \normF{P'-P}&\leq 
        \gamma\left(\normF{A}^{2}\right)\left(2\left(\frac{\kappa_{2}}{\normL{A}}\right)^{4}\left\{\phi+7\sqrt{2}\Omega\frac{\kappa_{2}}{\normL{A}}\right\}\right)\\
        &=2\gamma\normF{A}^{2}\left(\frac{\kappa_{2}}{\normL{A}}\right)^{4}\left\{\phi+7\sqrt{2}\Omega\frac{\kappa_{2}}{\normL{A}}\right\},
    \end{align*}
    as claimed.
\end{proof}

\begin{proof}[Proof of Proposition \ref{ovec}]
    Consider the same function $h\colon \Realnn\to\Realnn$ as in the proof of Lemma \ref{ppcl}. Remember that we have $\Phi_{h}(S^{*}S)=\Phi_{f}(S^{*})\Phi_{\inv}(S)$. As discussed in Section~\ref{sec-pre}, we also have $\Phi_{\inv}(S^{*})S=S^{*}\Phi_{\inv}(S)=\Pi_{\row(S)}$. We can thus write
    \[S^{*}PS=S^{*}(\Phi_{\inv}(S)\Phi_{f}(S^{*})\Phi_{\inv}(S)\Phi_{\inv}(S^{*}))S=\Pi_{\row(S)}\Phi_{h}(S^{*}S)\Pi_{\row(S)}=\Phi_{h}(S^{*}S).\]
    
    Similarly, observe that $\Phi_{h}(A^{*}A)A^{*}=\Phi_{f}(A^{*})\Pi_{\col(A)}=\Phi_{f}(A^{*})$. We can thus write:
    \[\begin{split}
        \nor{x-\Phi_{f}(A^{*})b} &= \nor{S^{*}P'SA^{*}b-(S^{*}PS)A^{*}b + \Phi_{h}(S^{*}S)A^{*}b-\Phi_{h}(A^{*}A)A^{*}b}\\
        &\leq \nor{S^{*}P'SA^{*}b-S^{*}PSA^{*}b}+\nor{\Phi_{h}(S^{*}S)A^{*}b-\Phi_{h}(A^{*}A)A^{*}b}\\
        &\leq \left(\normL{S^{*}}\normF{P'-P}\normL{S}+\normF{\Phi_{h}(A^{*}A)-\Phi_{h}(S^{*}S)}\right)\nor{A^{*}b}.
    \end{split}\]
    
    Using Lemma \ref{fpsd} and the definitions of set $L$ and $Q$ in Equation (\ref{eq:intL}), we get
    \[\begin{split}
        \nor{x-\Phi_{f}(A^{*})b}&\leq \normF{P'-P}\max_{\sigma\in L}\left\{\absl*{\sigma}^{2}\right\}\normL{A}\nor{b}\\
        &+\theta\normF{A}^{2}\left(\max_{\sigma\in Q}\left\{\absl*{h'(\sigma)}+\absl*{\frac{h(\sigma)}{\sigma}}\right\}\right)\normL{A}\nor{b}.\\
    \end{split}\]
    
    Now, similarly to the proof of Lemma \ref{ppcl}, we have
    \[\max_{\sigma\in Q}\left\{\absl*{h'(\sigma)}+\absl*{\frac{h(\sigma)}{\sigma}}\right\}\leq\left(\frac{\kappa_{2}}{\normL{A}}\right)^{2}\left\{\phi+3\sqrt{2}\Omega\frac{\kappa_{2}}{\normL{A}}\right\}.\]
    
    Using this inequality and Equation (\ref{eq:ppcl}), we get
    \[\begin{split}
        \nor{x-\Phi_{f}(A^{*})b}\\
        &\hspace{-18mm}\leq \left(2\gamma\normF{A}^{2}\left(\frac{\kappa_{2}}{\normL{A}}\right)^{4}\left\{\phi+7\sqrt{2}\Omega\frac{\kappa_{2}}{\normL{A}}\right\}\right)\frac{\normL{A}^{2}}{2\kappa_{2}^{2}}\normL{A}\nor{b}\\
        &\hspace{-18mm}+\theta\normF{A}^{2}\left(\frac{\kappa_{2}}{\normL{A}}\right)^{2}\left\{\phi+3\sqrt{2}\Omega\frac{\kappa_{2}}{\normL{A}}\right\}\normL{A}\nor{b}\\
        &\hspace{-18mm}\leq \left(\gamma\normF{A}^{2}\frac{\kappa_{2}^{2}}{\normL{A}}\left\{\phi+7\sqrt{2}\Omega\frac{\kappa_{2}}{\normL{A}}\right\}+\theta\normF{A}^{2}\frac{\kappa_{2}^{2}}{\normL{A}}\left\{\phi+3\sqrt{2}\Omega\frac{\kappa_{2}}{\normL{A}}\right\}\right)\nor{b}\\
        &\hspace{-18mm}= \frac{\epsilon}{2}+\frac{\epsilon}{2} = \epsilon.
    \end{split}\]
    
    Thus we obtain $\nor{x-\Phi_{f}(A^{*})b}\leq \epsilon$ when choosing the values for $\gamma$ and $\theta$ in the statement of the lemma (straightforward calculations show that for $\epsilon$ satisfying Inequality (\ref{eq:eps}) these values are in the ranges allowed for the parameters $\gamma$ and $\theta$ in Algorithm \ref{alg1}).
\end{proof}

\section{Complete analysis of Algorithm \ref{alg3}}\label{appC}
    Let us first show the correctness of Algorithm \ref{alg3}. Similarly to the analysis done in Theorem~\ref{th1}, with probability at least $1-\eta$, the vector $x'$ satisfies
    \[\nor{x'-\Phi_{f}(A^{*})b}\leq \frac{\epsilon_{2}\omega\alpha}{2}\nor{b}.\]
    
    Remember that the norm of the projection of $b$ on the column space of $\Phi_{f}(A)$ is at least $\alpha\nor{b}$.
    Consider $b=\sum_{i=1}^{m}b_{i}u_{i}$, where $u_{i}$ are the left singular vectors of matrix $\Phi_{f}(A)\in \Complex^{m\times n}$, where $k$ is the rank of this matrix. So the following inequality holds:
    \begin{equation}\label{eq:proj}
        \begin{split}
            \nor{\Phi_{f}(A^{*})b}^{2}&=\sum_{i=1}^{\min{(m,n)}}\absl{b_{i}f(\sigma_{i}(A))}^{2}\\
            &\geq \sum_{i=1}^{k}\absl{b_{i}f(\sigma_{i}(A))}^{2}\\
            &\geq \left\{\min_{i\in [k]}f(\sigma_{i}(A))\right\}^{2}\sum_{i=1}^{k}\absl*{b_{i}}^{2}\\
            &\geq \omega^{2}\alpha^{2}\nor{b}^{2}.
        \end{split}
    \end{equation}
    
    Thus using Inequality (\ref{eq:dist}) and (\ref{eq:proj}), the following inequality is true
    \[\nor{\mathcal{P}_{x'}-\mathcal{P}_{\Phi_{f}(A^{*})b}}_{TV}\leq \epsilon_{2}\frac{\omega\alpha\nor{b}}{\nor{\Phi_{f}(A^{*})b}}\leq \epsilon_{2}.\]
    
    Let us now analyze the complexity of Algorithm \ref{alg3}.
    Step 4 in Algorithm \ref{alg3} uses Lemma \ref{samp} and has $O(r^{2}C(S^{*},P'z))$ sample complexity and $O(r^{2}C(S^{*},P'z)\log^{2}{(nr)})$ time complexity, where
    \[C(S^{*},P'z)=\frac{\sum_{i=1}^{r}\nor{(P'z)_{i}S_{(.,i)}}^{2}}{\nor{S^{*}P'z}^{2}}\leq \frac{\left(\sum_{i=1}^{r}\absl{(P'z)_{i}}\nor{S_{(.,i)}}\right)^2}{\nor{S^{*}P'z}^{2}}.\]
    
    Using the Cauchy-Schwarz inequality we obtain
    \[C(S^{*},P'z)\leq \frac{\sum_{i=1}^{r}\absl{(P'z)_{i}}^{2}\sum_{i=1}^{r}\nor{S_{(.,i)}}^{2}}{\nor{S^{*}P'z}^{2}}=\frac{\nor{P'z}^{2}\normF{S}^{2}}{\nor{x'}^{2}}.
    \]
    
    Now using Equation (\ref{eq:cond1}), the bound $\normL{P'}\le\Omega\left(\frac{\sqrt{2}\kappa_{2}}{\normL{A}}\right)^{3}$, the inequality $\nor{z-SA^{*}b}\leq \epsilon''\sqrt{r}$ and then Equation (\ref{eq:proj}), we obtain:
    \[\begin{split}
        C(S^{*},P'z)&\leq
        \frac{8\frac{\Omega^{2}\kappa_{2}^{6}}{\normL{A}^{6}}\nor{z}^{2}\normF{A}^{2}}{(1-\epsilon_{2}/2)^{2}\omega^{2}\alpha^{2}\nor{b}^{2}}\\
        &\leq\frac{8\frac{\Omega^{2}\kappa_{2}^{6}}{\normL{A}^{6}}\left(\nor{SA^{*}b}+\epsilon''\sqrt{r}\right)^{2}\normF{A}^{2}}{(1-\epsilon_{2}/2)^{2}\omega^{2}\alpha^{2}\nor{b}^{2}}\\
        &\leq \frac{8\frac{\Omega^{2}\kappa_{2}^{6}}{\normL{A}^{6}}\left(\frac{\normL{A}}{\sqrt{2}\kappa_{2}}\left(2\kappa_{2}^{2}+1\right)^{1/2}\normL{A}\nor{b}+\frac{\epsilon_{2}\omega\alpha\nor{b}}{8\Omega\sqrt{\left(2\kappa_{2}^{2}+1\right)}}\left(\frac{\normL{A}}{\kappa_{2}}\right)^{2}\right)^{2}\normF{A}^{2}}{(1-\epsilon_{2}/2)^{2}\omega^{2}\alpha^{2}\nor{b}^{2}}.
    \end{split}\]
    Using the bounds from Statement (\ref{eq:cond4}) and neglecting terms with $\epsilon_{2}$, we can write
    \[\begin{split}
        C(S^{*},P'z)&=
        O\left(\frac{\frac{\Omega^{2}\kappa_{2}^{6}}{\normL{A}^{6}}\left(\frac{\normL{A}^{2}}{2\kappa_{2}^{2}}\left(2\kappa_{2}^{2}+1\right)\normL{A}^{2}\nor{b}^{2}\right)\normF{A}^{2}}{\omega^{2}\alpha^{2}\nor{b}^{2}}\right)\\
        &= O\left(\frac{\kappa_{2}^{4}}{\alpha^{2}\normL{A}^{2}}\frac{\Omega^{2}}{\omega^{2}}\normF{A}^{2}(2\kappa_{2}^{2}+1)\right) = O\left(\frac{\kappa_{2}^{6}}{\alpha^{2}\normL{A}^{2}}\frac{\Omega^{2}}{\omega^{2}}\normF{A}^{2}\right).
    \end{split}\]
    The complexity of Step 4 in Algorithm \ref{alg3} thus dominates the sample complexity.
    The time complexity, on the other hand, is still dominated by the computation of the singular value decomposition of matrix $W$, as in Algorithm \ref{alg2}.
\end{document}